\documentclass[11pt,reqno]{amsart}

\usepackage{amssymb}
\usepackage{amsmath}
\usepackage{mathrsfs}
\usepackage{amsthm}

\usepackage{changes}
\newcommand{\stkout}[1]{\ifmmode\text{\sout{\ensuremath{#1}}}\else\sout{#1}\fi}
\setdeletedmarkup{\stkout{#1}}
\makeatletter 
\@namedef{Changes@AuthorColor}{red}
\colorlet{Changes@Color}{red}
\makeatother

\usepackage{color}

\usepackage{algorithm}
\usepackage{algorithmic}

\newtheorem{Theorem}{Theorem}[section]
\newtheorem{Proposition}[Theorem]{Proposition}
\newtheorem{Lemma}[Theorem]{Lemma}
\newtheorem{Corollary}[Theorem]{Corollary}
\newtheorem{Conjecture}[Theorem]{Conjecture}

\theoremstyle{definition}
\newtheorem{Definition}[Theorem]{Definition}
\newtheorem{Remark}{Remark}[section]
\newtheorem{Problem}{Problem}[section]
\newtheorem{Example}{Example}[section]

\usepackage{geometry}
\usepackage{graphicx}

\usepackage{hyperref}
\hypersetup{colorlinks=true,
            linkcolor=blue,
            anchorcolor=blue,
            citecolor=blue}
            
\usepackage{changes}

\newcommand{\FF}{\mathbb{F}}
\newcommand{\ZZ}{\mathbb{Z}}
\newcommand{\CC}{\mathbb{C}}

\makeatletter
\newenvironment{breakablealgorithm}
  {% \begin{breakablealgorithm}
   \begin{center}
     \refstepcounter{algorithm}% New algorithm
     \hrule height.8pt depth0pt \kern2pt% \@fs@pre for \@fs@ruled
     \renewcommand{\caption}[2][\relax]{% Make a new \caption
       {\raggedright\textbf{\ALG@name~\thealgorithm} ##2\par}%
       \ifx\relax##1\relax % #1 is \relax
         \addcontentsline{loa}{algorithm}{\protect\numberline{\thealgorithm}##2}%
       \else % #1 is not \relax
         \addcontentsline{loa}{algorithm}{\protect\numberline{\thealgorithm}##1}%
       \fi
       \kern2pt\hrule\kern2pt
     }
  }{% \end{breakablealgorithm}
     \kern2pt\hrule\relax% \@fs@post for \@fs@ruled
   \end{center}
  }
\makeatother

\makeatletter 
\def\@eqnnum{{\normalfont \color{blue} (\theequation)}} 
\makeatother

\begin{document}

\title{Improving the Gilbert-Varshamov Bound by Graph Spectral Method}
\author{Zicheng Ye$^{1,2}$, Huazi Zhang${^3}$, Rong Li${^3}$, Jun Wang${^3}$, Guiying Yan$^{1,2*}$ and Zhiming Ma$^{1,2*}$} 
\address{\noindent $^1$Academy of Mathematics and Systems Science, CAS\newline
$^2$University of Chinese Academy of Sciences\newline
$^3$Huawei Technologies Co. Ltd}
\email{yezicheng@amss.ac.cn, zhanghuazi@huawei.com, lirongone.li@huawei.com, wangjun@huawei.com, yangy@amss.ac.cn, mazm@amt.ac.cn}


\begin{abstract}
We improve Gilbert-Varshamov bound by graph spectral  method. Gilbert graph $G_{q,n,d}$ is a graph with all vectors in $\FF_q^n$ as vertices where two vertices are adjacent if their Hamming distance is less than $d$. In this paper, we calculate the eigenvalues and eigenvectors of $G_{q,n,d}$ using the properties of Cayley graph. The improved bound is associated with the minimum eigenvalue of the graph. Finally we give an algorithm to calculate the bound and linear codes which satisfy the bound. 
\end{abstract}

\maketitle

\section{Introduction}

Let $q$ be a prime number and $\FF_q$ be the finite field given by the integers mod $q$.  $\FF_q^n$ is the $n$-dimension vector space over $\FF_q$. A subset $C$ of $\FF_q^n$ is called a $q$-ary code with length $n$. $C$ is said to be linear if $C$ is a subspace. The vectors in $C$ are called codewords. The dimension of $C$ is given by $k=\log_q|C|$, and the rate is given by $k/n$. 

Let $c=(c_1,...,c_n)$ be a vector in $\FF_q^n$. The Hamming weight of $c$ is $w(c) = |\{i|c_i \neq 0\}|$. The Hamming distance between two vectors $c, c'\in \FF_q^n$ is $d(c,c') = |\{i|c_i \neq c_i'\}|$. $C$ is called a code with minimum distance $d$ if the distance of any two distinct codewords in $C$ are greater or equal to $d$. The relative distance of $C$ is then given by $d/n$. A code in $\FF_q^n$ with dimension $k$ and minimum distance $d$ is called an $[n,k,d]_q$ code.

Let $A_q(n, d)$ be the maximum number of codewords in a $q$-ary code with length $n$ and minimum Hamming distance $d$. Finding the value of $A_q(n,d)$ is a very fundamental and difficult problem in coding theory\cite{sloane1989unsolved}. The first and most important lower bound of $A_q(n,d)$ is Gilbert-Varshamov bound:

\begin{Proposition}[Gilbert-Varshamov Bound\cite{gilbert1952comparison}]\label{pro:1.1}
Let
\[V_q(n,d)=\sum_{i=0}^d \binom{n}{i}(q-1)^i,\]
be the number of vectors with Hamming weight less than $d$, then
\begin{equation}
A_q(n, d)\geq \frac{q^n}{V_q(n,d-1)}.
\end{equation}
\end{Proposition}

Proposition \ref{pro:1.1} has been improved variously in \cite{varshamov1957estimate}, \cite{elia1983some},\cite{barg2000strengthening}, \cite{trachtenberg2002designing},\cite{jiang2004asymptotic}, \cite{o2006bounds}, \cite{spasov2009some} and\cite{hao2020distribution}. Among them, the best improvement on the order of magnitude is from Jiang and Vardy \cite{jiang2004asymptotic} by studying the independence number of the graph $G_{q,n,d}$ defined as follows:

\begin{Definition}\cite{jiang2004asymptotic}\label{def:2.1}
Gilbert graph $G_{q,n,d}$ is a graph whose $V(G_{q,n,d}) = \FF_q^n$ and $\forall u,v \in V(G_{q,n,d})$, $(u,v)\in E(G_{q,n,d})$ if and only if $1 \leq d(u, v)\leq d- 1$.
\end{Definition}

People are also interested to the asymptotic form of Gilbert-Varshamov bound as $n$ goes to infinity. The maximum rate of code families with  relative distance $\delta$ is defined as
\[\beta_q(\delta)=\lim_{n\to\infty} \frac{1}{n} A_q(n,n\delta).\]
Notice that
\[\frac{1}{n}\log_q V_q(n, d) = h_q(\frac{d}{n})+o(1)\]
as $n\to\infty$ where
\[h_q(x)=x\log_q(q-1)-x\log_qx-(1-x)\log_q(1-x).\]

This implies the asymptotic form of Proposition \ref{pro:1.1}:

\begin{Proposition}[Asymptotic Gilbert-Varshamov Bound]\label{pro:1.2}
For every $0\leq \delta<1-1/q$, 
\begin{equation}
\beta_q(\delta)\geq 1-h_q(\delta).
\end{equation}
\end{Proposition}

Tsfasman, Vladuts and Zink \cite{tsfasman1982modular} have proved that $\beta_q(\delta)> 1-h_q(\delta)$ for some $q \geq 49$. However, when $q = 2$, some people conjecture that there does not exist any binary code with relative distance $\delta$ and rate $R>1-h_2(\delta)$ as $n\to\infty$. \cite{jiang2004asymptotic}

In this paper, we improve Gilbert-Varshamov bound by analysing the spectrum of $G_{q,n,d}$ and bounding its independence number. In section \ref{sec:2}, we use properties of Cayley graph to calculate the closed form of the eigenvalues and eigenvectors of $G_{q,n,d}$, and get an upper bound and a lower bound using the minimum eigenvalue. The improvement of Gilbert-Varshamov bound is given in section \ref{sec:3}. Section \ref{sec:4} concludes the results and gives some open problems.

\section{the Spectrum of the Gilbert Graph}\label{sec:2}

A graph $G$ is a set of $N$ vertices $V(G)=\{v_1,...,v_N\}$ and a set of edges $E(G)\subseteq V(G)\times V(G)$. For $v_i,v_j\in V(G)$, we say $v_i$ and $v_j$ are adjacent if $(v_i,v_j)\in E(G)$. The neighbourhood of $v\in V(G)$ is the set of all vertices adjacent to $v$ and denoted by $N(v)$. The number of edges that are incident to $v$ is the degree of $v$. A graph is called a $D$-regular graph if each vertex has degree $D$. The adjacent matrix for $G$ is a $N\times N$ matrix $A$ where $A_{ij}$ is 1 if $v_i$ and $v_j$ are adjacent, and 0 otherwise. 

A subset of $V$ is called an independent set if none of vertices in the set are adjacent. The independence number of $G$ is the size of the largest independent set in $G$ and denoted by $\alpha(G)$. 

The spectrum of a graph is the eigenvalues of its adjacent matrix, which has a strong relationship with the structure of the graph \cite{brouwer2011spectra}, including independence number we study in this paper.

It is clear that $G_{q,n,d}$ is a $(V_q(n,d - 1) - 1)$-regular graph. More exactly, $G_{q,n,d}$ is vertex-transitive\cite{el2012graph}, \cite{godsil2013algebraic}. Then Gilbert-Varshamov bound is the direct consequence of following facts:

\begin{Proposition}\label{prop:2.1}
\[\alpha(G_{q,n,d}) = A_q(n,d).\]
\end{Proposition}

\begin{Lemma}\label{lemma:1.5}
\[\alpha(G)\geq\frac{|V(G)|}{\Delta(G)+1}\]
where $\Delta(G)$ is the maximal degree of $G$.
\end{Lemma}

Here, Proposition \ref{prop:2.1} is from that a $q$-ary code of length $n$ has minimum distance $d$ if and only if it is an independent set in $G_{q,n,d}$. Lemma \ref{lemma:1.5} holds since if some vertex $v$ is in an independent set $I$ of graph $G$, it forbids at most $\Delta(G)+1$ vertices (including $v$ itself) to be added into $I$.

\begin{Definition}[\cite{godsil2013algebraic}]
A group $(H,\cdot)$ is a set $H$ together with a binary operation $H\times H\to H$ denoted as $\cdot$ such that the following properties are satisfied:

1. For all $a, b, c \in G$, $(a\cdot b)\cdot c = a\cdot(b\cdot c)$. 

2. There exists an element $e\in H$ called identity element such that for every $a\in H$, $a\cdot e = e\cdot a=a$. 

3. For each $a \in H$, there exists a unique element $a^{-1} \in H$ called inverse element of $a$ such that $a\cdot a^{-1} = a^{-1}\cdot a=e$. 

Moreover, if a group $(H,\cdot)$ also satisfies $a\cdot b = b\cdot a$ for all $a, b \in H$, it is called an Abelian group.
\end{Definition}

\begin{Definition}[\cite{godsil2013algebraic}, \cite{brouwer2011spectra}, \cite{apostol2013introduction}]
Denote that content $(H,\cdot)$ is a finite group and $S \subseteq H$ satisfies $\{s^{-1}|s\in S\}=S$ and $1\notin S$. The (finite and undirected) Cayley graph on $H$ with difference set $S$ is denoted as $\Gamma$ with vertex set $H$ and edge set $E = \{(x, y) | yx^{-1} \in S\}$. 
\end{Definition}

In this paper, we always assume $H$ is an Abelian group.

\begin{Example}
$(\FF_q^n,+)$ is an Abelian group with $(u_1,u_2,...,u_n)+(v_1,v_2,...,v_n)=(u_1+v_1,u_2+v_2,...,u_n+v_n)$, here '$+$' is the addition operation of integers mod $q$. The identity element is $(0,0,...,0)$ and the inverse element of $u=(u_1,u_2,...,u_n)$ is $-u=(-u_1,-u_2...,-u_n)$.

Since $w(u)$=$w(-u)$, $G_{q,n,d}$ can be regarded as a Cayley graph on $H= \mathbb{F}_q^n$ with $S=\{u\in \mathbb{F}_q^n|1\leq w(u)\leq d-1\}$.
\end{Example}

To compute the spectrum of a Cayley graph $\Gamma$, let us define a character of $H$ to be a map $\chi : H \to \CC^*$ satisfying $\chi(xy) = \chi(x)\chi(y)$, here $\CC$ is the set of complex numbers and $\CC^*=\CC-\{0\}$. Then
\[\sum_{y\in N(x)}\chi(y) = (\sum_{s\in S}\chi(s))\chi(x),\]
so the vector $(\chi(x))_{x\in H}$ is an eigenvector of the adjacency matrix of $\Gamma$ with eigenvalue $\sum_{s\in S}\chi(s)$.

To calculate the closed form of the spectrum for $G_{q,n,d}$, we need to use Krawtchouk polynomials. 

\begin{Definition}[\cite{levenshtein1995krawtchouk}]
For positive integers $q,k,n$, Krawtchouk polynomials are defined as
\begin{equation}K_k(x;n,q) = \sum_{j=0}^k (-1)^j \binom{x}{j} \binom{n-x}{k-j}(q-1)^{k-j},
\end{equation}
where 
\[\binom{x}{j}=\frac{x(x-1)...(x-j+1)}{j!}.\]
Specially, $\binom{x}{0} = 1$.
\end{Definition}

\begin{Lemma}[\cite{levenshtein1995krawtchouk}]\label{lemma:2.4}
When $x = 1,...,n$, 
\begin{equation}
\sum_{k=0}^{d-1} K_k(x;n,q)= K_{d-1}(x-1;n-1,q),
\end{equation}

\end{Lemma}

It is well-known that Delsarte\cite{delsarte1973algebraic} has used Krawtchouk polynomials to prove an upper bound of $A_q(n,d)$ by linear programming. We now want to use them to prove a lower bound by spectral graph theory.

Write $u=(u_1,....,u_n), v=(v_1,...,v_n) \in \FF_q^n$ and define $\langle u,v\rangle=\sum_{i=1}^n u_iv_i\in \FF_q$. Now we have the following theorem on the spectrum of $G_{q,n,d}$. 

\begin{Theorem}\label{thm:2.5}
Denote $z=\exp\{\frac{2\pi i}{q}\}\in\CC$. The $q^n$ orthogonal eigenvectors of $G_{q,n,d}$ are 
\begin{equation}
a_v  =(z^{\langle u,v\rangle})_{u\in V(G_{q,n,d})}
\end{equation}
for all $v\in \FF_q^n$. The corresponding eigenvalue $\lambda_v$ of $a_v$ is $K_{d-1}(w(v)-1,n-1,q)-1$ if $w(v)\neq 0$ and $V_q(n,d-1)-1$ if $w(v)= 0$. 
\end{Theorem}

\begin{proof}

Recall $G_{q,n,d}$ is a Cayley graph on $H= \FF_q^n$ with difference set $S=\{u\in \FF_q^n|1\leq w(u)\leq d-1\}$. Let $\chi$ be a character of $(\mathbb{F}_q^n,+)$ and $e_i$ be the vector in $\FF_q^n$ with $1$ on $i$-th position and $0$ otherwise. For any $u=(u_1,...,u_n)\in \mathbb{F}_q^n$, 
\[\chi(u) = \chi(\sum_{i=1}^n u_ie_i) = \prod_{i=1}^n \chi(e_i)^{u_i},\]
so $\chi(u)$ can be determined by $\{\chi(e_i)\}_{1\leq i\leq n}$. Since $\chi(qe_i)=\chi(e_i)^q=1$, $\chi(e_i)$ must be a $q$-th root of unity. For all $v=(v_1,...,v_n)\in\FF_q^n$, define $\chi_v$ to be a character satisfying $\chi_v(e_i) = z^{v_i}$. The $q^n$ eigenvectors of $G_{q,n,d}$ are 
\[a_v=(\chi_v(u))_{u\in V(G_{q,n,d})}=(\prod_{i=1}^n \chi_v(e_i)^{u_i})_{u\in V(G_{q,n,d})}=(z^{\langle u,v\rangle})_{u\in V(G_{q,n,d})}.\]

For any $v,v'\in\FF_q^n$ and $v\neq v'$, 
\[\sum_{u\in \FF_q^n}\chi_v(u)\overline{\chi_{v'}(u)}=\sum_{u\in \FF_q^n}z^{\langle u,v-v'\rangle}=\prod_{i=1}^n\sum_{j=0}^{q-1}z^{j(v_i-v'_i)}=0. \]

Therefore, these $q^n$ eigenvectors are orthogonal with each other, and hence linearly independent.

The eigenvalues of $a_v$ are $\lambda_v = \sum_{u\in S}\chi_v(u)$. Let $supp(v) = \{i|v_i \neq 0\}$. If $w(v)\neq 0$, for any integers $k,j\in\ZZ$, denote 
\[\mathcal{A}_{k,j}=\{A \subseteq \{1,...,n\}\big| |A|= k \text{ and } |A\cap supp(v)|=j\},\]
then $|\mathcal{A}_{k,j}|=\binom{n-w(v)}{k-j}\binom{w(v)}{j}$ and $S=\bigcup_{k=1}^{d-1}\bigcup_{j=0}^{k}\mathcal{A}_{k,j}$. For any $A\in \mathcal{A}_{k,j}$,
\begin{align}
\notag
\sum_{supp(u)=A}\chi_v(u)&=\sum_{supp(u)=A}z^{\langle u,v\rangle}\\ \notag
&=\prod_{i\in A}(z^{v_i}+z^{2v_i}+...+z^{(q-1)v_i})=(q-1)^{k-j}(-1)^j. 
\end{align}

Therefore, 
\begin{align}
\notag
\sum_{u\in S}\chi_v(u)&=\sum_{k=1}^{d-1}\sum_{j=0}^k\sum_{A\in\mathcal{A}_{k,j}} \sum_{supp(u)=A}\chi_v(u)\\
\notag
&=\sum_{k=1}^{d-1}\sum_{j=0}^k \binom{n-w(v)}{k-j}\binom{w(v)}{j}(q-1)^{k-j}(-1)^j\\ 
\notag
&=\sum_{k=1}^{d-1}K_k(w(v);n,q) = K_{d-1}(w(v)-1;n-1,q)-1. 
\end{align}
Here the last equation is from Lemma \ref{lemma:2.4}. 

If $w(v)=0$, $\chi_v(u)=1$ for all $u\in \FF_q^n$, so $\lambda_v = \sum_{u\in S}\chi_v(u)=V_q(n,d-1)-1$, which is the degree of $G_{q,n,d}$.
\end{proof}

From the proof, we see the maximum eigenvalue of $G_{q,n,d}$ is $V_q(n,d-1)-1$ and the corresponding eigenvector is all-one vector $\textbf{1}$. 

Let $\lambda_{min}(G)$ be the minimum eigenvalue of $G$, then $\lambda_{min}(G_{q,n,d})=\min_{1\leq x\leq n} K_{d-1}(x-1;n-1,q)-1$. It seems not easy to find the closed form of $x$ which minimize $K_{d-1}(x-1;n-1,q)$, but we can compute the accurate value of $\lambda_{min}(G_{q,n,d})$ in polynomial time for given $n,q$ and $d$ by Theorem \ref{thm:2.5}.

Now we can apply two bounds on independence number related to the minimum eigenvalue. The famous Hoffman bound for $G_{q,n,d}$ is:

\begin{Proposition}\label{pro:2.8}
\begin{equation}
\alpha(G_{q,n,d})\leq \frac{-q^n\lambda_{min}(G_{q,n,d})}{V_q(n,d-1)-1+\lambda_{min}(G_{q,n,d})}.
\end{equation}
\end{Proposition}

The second bound is from Wilf\cite{wilf1986spectral}:

\begin{Proposition}\label{pro:2.7}
Let G be a $D$-regular graph of $N$ vertices, then
\[\alpha(G)\geq \frac{N}{D+1+M(\lambda_{min}(G)+1)/N},\]
where 
\[M=\max\{M_+^2,M_-^2\},\]
\[M_+ = \min_{b(G)_i>0}\frac{1}{b(G)_i},\ M_- = \min_{b(G)_i<0}-\frac{1}{b(G)_i},\]
and $b(G)$ is one of normalized real eigenvectors of the minimum eigenvalue.

\end{Proposition}

To apply Proposition \ref{pro:2.7}, we need real eigenvectors of $G_{q,n,d}$, but when $q\geq 3$, the eigenvectors given from Theorem \ref{thm:2.5} are complex. By our knowledge now, we can only get part of real eigenvectors from the following idea.
 
Given a nonempty set $A\subseteq\{1,...,n\}$ , for $j=1,...,n-1$, define $v^{(j)}\in \FF_q^n$ by
\[(v^{(j)})_i=\left\{\begin{array}{cc}
z^j & \text{ if } i\in A; \\
0 & \text{ otherwise }
\end{array}\right.\]
Define $b_A=\sum_{j=1}^{n-1} a_{v^{(j)}}$. Notice that the value of $\lambda_v$ is determined by $w(v)$ and $w(v^{(j)})=|A|$ for all $j$, then $b_A$ is an eigenvector of $G_{q,n,d}$ with eigenvalue $K_{d-1}(|A|-1;n-1,q)-1$. Since the entries of $b_A$ are only $-1$ and $q-1$ and $b_A^T\textbf{1}=0$, we know that $||b_A||_2=\sqrt{q^n(q-1)}$. 

There exists some  $A\subseteq \{1,...,n\}$ such that $b(G_{q,n,d})=b_A/||b_A||_2$ and $M = q^n(q-1)$ in Proposition \ref{pro:2.7}. Therefore,
 
\begin{Corollary}\label{cor:2.7}
\begin{equation}
\alpha(G_{q,n,d})\geq \frac{q^n}{V_q(n,d-1)+(q-1)\lambda_{min}(G_{q,n,d})+q}.
\end{equation}
\end{Corollary}

We notice that corollary \ref{cor:2.7} is an improvement of Gilbert-Varshamov bound since $\lambda_{min}(G_{q,n,d})<-1$. Proposition \ref{pro:2.8} and Corollary \ref{cor:2.7} can be calculated in polynomial time for given $n,q$ and $d$ by Theorem \ref{thm:2.5}.

\section{Improved Lower Bound}\label{sec:3}

In this section, we will improve the lower bound further.

Let $G_0 = G_{q,n,d}$. $\lambda_{min}^{(0)}$ is the minimum eigenvalue of $G_0$ and $a_{v^{(0)}}$ is one of corresponding eigenvectors. By induction, assume we already have a series of orthogonal vectors $v^{(0)},...,v^{(t-1)}$. Define $G_t$ to be the induced graph of $G_0$ with vertex set $V_t=\bigcup_{i=0}^{t-1}\{u\in \FF_q^n|\langle u,v^{(i)}\rangle=0\}$ and edge set $E(G_0)\cap (V_t\times V_t)$. Notice that $V_t$ is a subspace of $\FF_q^n$. 

\begin{Lemma}\label{lemma:3.1}
For $0\leq t< n$, $G_t$ is Cayley graph with $q^{n-t}$ vertices and difference set $S_t = V_t\cap \{u\in \mathbb{F}_q^n|1\leq w(u)\leq d-1\}$. The $q^{n-t}$ orthogonal eigenvectors of $G_t$ are 
\begin{equation}
a_v^{(t)} = (\chi_v(u))_{u\in V(G_t)}= (z^{\langle u,v\rangle})_{u\in V(G_t)},
\end{equation}
where $v$ is not a linear combination of $v^{(0)},...,v^{(t-1)}$ and $\chi_v$ is defined in Theorem \ref{thm:2.5}.
\end{Lemma}

\begin{proof}
Notice that $a_v^{(t)}=(z^{\langle u,v\rangle})_{u\in V(G_t)}$ and $a_{v'}^{(t)}=(z^{\langle u,v'\rangle})_{u\in V(G_t)}$ are the same vector on $V_t$ if and only if $v-v'$ is a linear combination of $v^{(0)},...,v^{(t-1)}$, so we only need to  consider $q^{n-t}$ eigenvectors  $a_v^{(t)}$ where $v$ is orthogonal to $v^{(0)},...,v^{(t-1)}$. Now we need to prove that these $q^{n-t}$  eigenvectors are orthogonal with each other. Suppose $v-v'$ is not a linear combination of $v^{(0)},...,v^{(t-1)}$, then there must be some $w\in V_t$ which is not orthogonal to $v-v'$. Expand $w_1$ into $\{w_1,...,w_{n-t}\}$ which is a basis for $V_t$. Therefore,
\[\sum_{u\in V_t} z^{\langle u,v-v'\rangle} = \prod_{i=1}^{n-t}\sum_{j=0}^{q-1} z^{j\langle w_i,v-v'\rangle}=0. \qedhere\]
\end{proof}

Let $\lambda_{min}^{(t)}$ denote the minimum eigenvalue of $G_{t}$ and $D^{(t)}$ denote the degree of vertices in $G_t$. Note that the dimension of the eigenspace of $\lambda_{min}^{(t)}$ may be larger than one. For convenience, we can regard the vectors in $\FF_q^n$ as base-$q$ numbers and choose $v^{(t)}$ as the smallest one such that $a^{(t)}_{v^{(t)}}$ is the eigenvector of $\lambda_{min}^{(t)}$. Now define $V_{t+1}=V_t \cup \{u\in\FF_q^n|(a_{v^{(t)}})_u=1\}=\bigcup_{i=0}^{t}\{u\in\FF_q^n|\langle u,v^{(i)}\rangle=0\}$ and $G_{t+1}$ is the induced graph from $V_{t+1}$.

\begin{Lemma}\label{lemma:3.2}
The degree of every vertex in $G_{t+1}$ is $D^{(t+1)} = (D^{(t)}+(q-1)\lambda_{min}^{(t)})/q$. 
\end{Lemma}

\begin{proof}

Recall $z = \exp\{2\pi i/q\}\in\CC$. To see the degree of $G_{t+1}$, we notice that for any $j=1,...,q-1$, $u\in S_t$ if and only if $iu\in S_t$. Also,
\[\sum_{j=1}^{q-1} z^{\langle ju,v^{(t)}\rangle} = \left\{\begin{array}{cc}
q-1, &\text{ if } \langle u,v^{(t)}\rangle  = 0\\
-1, &\text{ otherwise }
\end{array}\right.\]
Therefore, 
\[\lambda_{min}^{(t)} = \sum_{u\in S_{t}}z^{\langle u,v^{(t)}\rangle}=|S_{t+1}|-\frac{1}{q-1}(|S_t|-|S_{t+1}|)  = D^{(t+1)}-\frac{1}{q-1}(D^{(t)}-D^{(t+1)}) \] 
where the second equality follows from $S_{t+1} = S_t\cap V_{t+1} = S_t\cap \{u\in V_t|\langle u,v^{(t)}\rangle = 0\}$.
\end{proof}

\begin{Remark}\label{re:3.1}
From Lemma \ref{lemma:3.1} and \ref{lemma:3.2}, we can say that $G_{t+1}$ is sparser than $G_{t}$. In explicit language, the ratio of vertex number to degree of $G_{t+1}$ is larger than that of $G_{t}$, since
\[\frac{|V(G_{t+1})|}{D^{(t+1)}+1}= \frac{|V(G_t)|}{D^{(t)}+(q-1)\lambda_{min}^{(t)}+q}\geq \frac{|V(G_{t})|}{D^{(t)}+1}\] 
if $\lambda_{min}^{(t)}\leq -1$, which holds for any graph with at least one edge. This means that the lower bound in Lemma \ref{lemma:1.5} for $G_{t+1}$ is better than that for $G_t$. In fact, among all subgraphs of $G_t$ whose vertex set is a subgroup with $V_t/q$ elements, $G_{t+1}$ is the sparsest one.
\end{Remark}

Due to Remark \ref{re:3.1}, we can improve the Gilbert-Varshamov bound:

\begin{Theorem}\label{thm:3.2}
If $G_{t+1}$ has at least one edge, then
\begin{equation}
A_q(n,d)=\alpha(G_{q,n,d})\geq \frac{q^n}{V_q(n,d-1)+\sum_{i=0}^{t}(q-1)q^i\lambda_{min}^{(i)}+q^{t+1}}.\\
\end{equation}
\end{Theorem}

\begin{proof}
From Lemma \ref{lemma:3.2}, 
\[D^{(t+1)} = \frac{D^{(t)}+(q-1)\lambda_{min}^{(t)}}{q}=\frac{V_q(n,d-1)+\sum_{i=0}^{t}(q-1)q^i\lambda_{min}^{(i)}}{q^{t+1}},\]
then
\[\alpha(G_{q,n,d})\geq \alpha(G_{t+1})\geq \frac{|V(G_{t+1})|}{D^{(t+1)}}=\frac{q^n}{V_q(n,d-1)+\sum_{i=0}^{t}(q-1)q^i\lambda_{min}^{(i)}+q^{t+1}}.\qedhere\]
\end{proof}

Theorem \ref{thm:3.2} is exactly Corollary \ref{cor:2.7} when $t=0$, and improves Corollary \ref{cor:2.7} when $t\geq 1$.

We can repeat the above procedures to get new subgraphs and improve the bound until for some integer $s$ the graph $G_s$ has no edges, so $\lambda_{min}^{(s)}=0$. Now $V_s$ is an independent set of $G_{q,n,d}$ and a subspace of $\FF_q^n$, so $V_s$ is a $[n,n-s,d]_q$ linear code. Therefore, 

\begin{Theorem}\label{thm:3.3}
If $\lambda_{min}^{(s)}=0$, then there is a $[n,n-s,d]_q$ linear code. 
\end{Theorem}

This implies that Theorem \ref{thm:3.2} also holds for linear codes.

In other words, $\lambda_{min}^{(s)}=0$ is equal to $D^{(s)}=0$, which implies 
\[V_q(n,d-1)+\sum_{t=0}^{s-1}(q-1)q^t\lambda_{min}^{(t)} = 0,\]
then 

\begin{Theorem}
For any sequence $\{b_t\}$ with $b_t\geq \lambda_{min}^{(t)}$ and $0\leq t\leq s-1$, there exists a $[n,k,d]_q$ linear code if
\[V_q(n,d-1)+\sum_{t=0}^{n-k-1}(q-1)q^tb_t\geq 0.\]
\end{Theorem}

Since $V_s= \bigcup_{i=0}^{s-1}\{u\in \FF_q^n|\langle u,v^{(i-1)}\rangle=0\}$, the code with parity check matrix $(v^{(0)},...,v^{(s-1)})$  satisfies Theorem \ref{thm:3.3}. 

Now the last problem is how to calculate $\lambda_{min}^{(t)}$. For any $v\in\mathbb{F}_q^n$, the eigenvector of $G_t$ is $a^{(t)}_v=(z^{\langle u,v\rangle})_{u\in V_t}$, and the eigenvalue is $\lambda^{(t)}_v=\sum_{u\in S_t}z^{\langle u,v\rangle}$. Now define
\[v_r=v+rv^{(t-1)}.\]
Then 
\[\sum_{r=0}^{q-1}\chi_{v_r}(u) = z^{\langle u,v\rangle}\sum_{r=0}^{q-1}z^{r\langle u,v^{(t)}\rangle} = q\chi_v(u)\mathbf{1}_{\langle u,v^{(t-1)}\rangle = 0}.\]
Thus,
\[\lambda^{(t)}_v=\sum_{u\in S_t}\chi_v(u)=\sum_{u\in S_{t-1}}\chi_v(u)\mathbf{1}_{\langle u,v^{(t-1)}\rangle = 0}=\frac{1}{q}\sum_{r=0}^{q-1}\sum_{u\in S_{t-1}}\chi_{v_r}(u)=\frac{1}{q}\sum_{r=0}^{q-1}\lambda^{(t-1)}_{v_r}.\]

Since $\lambda^{(0)}_v$ is known by Theorem \ref{thm:2.5}, all eigenvalues of $G_t$ (including the minimum eigenvalue $\lambda_{min}^{(t)}$) can be obtained. Now we can complete the whole procedures. Algorithm 1 describes the whole process briefly.

\begin{breakablealgorithm}
\caption{Framework of Generating Subgraph and Linear Code}
\begin{algorithmic}\label{alg:1}
\STATE $G=G_{q,n,d}$, $t=0$; 
\REPEAT 
\STATE Calculate the minimum eigenvalue $\lambda^{(t)}_{min}$ of $G$ and choose one of the corresponding eigenvector $a^{(t)}_{v^{(t)}}$; 
\STATE $G \Leftarrow$ the induced graph with vertex set $\{u\in V(G)|\langle u,v^{(t)}\rangle=0\}$;
\STATE $t \Leftarrow t+1$;
\UNTIL{$\lambda^{(t)}_{min}=0$}
\RETURN  the code with check matrix $(v^{(0)},...,v^{(t-1)})$.
\end{algorithmic}
\end{breakablealgorithm}

\begin{Remark}
Our results are proved when $q$ is a prime number. In fact, when $q$ is a prime power, the results still hold with little difference. The only difficult part is that the closed form for eigenvalues of $G_{q,n,d}$ will be much more complicated.
\end{Remark}

\section{Conclusions and Open Problems}\label{sec:4}

In this paper we use graph spectral method to improve Gilbert-Varshamov bound. The improvement is non-asymptotic. A natural question is to ask

\begin{Problem}
What is the asymptotic form of Theorem \ref{thm:3.2} or \ref{thm:3.3}?
\end{Problem}

We also design Algorithm \ref{alg:1} to give codes satisfying our improved bound. The time complexity of the algorithm is $O(nq^n)$. 
 
\begin{Problem}
Could it be possible to find algorithms with lower complexity?
\end{Problem}

In the algorithm, we use the minimum eigenvalue and the corresponding eigenvector of $G_t$ to construct $G_{t+1}$. It is possible to use other eigenvalues and eigenvectors instead of $\lambda^{(t)}_{min}$ and $a^{(t)}_{v^{(t)}}$ in Algorithm \ref{alg:1}. 

\begin{Conjecture}
For all vectors $v$ such that $a^{(t)}_v$ is the eigenvector with eigenvalue $\lambda_{min}^{(t)}$ in $G_t$, the graphs induced by $\{u\in V_t|\langle u,v\rangle=0\}$ are isomorphic.
\end{Conjecture}

The conjecture is from the symmetry of the graphs. However, we also want to know

\begin{Problem}
Is there a rule to choose $v^{(t)}$ such that $(v^{(0)},...,v^{(s-1)})$ has a good structure which is helpful for encoding and decoding?
\end{Problem}

However, if people choose other eigenvalues rather than the minimum one to construct $G_t$, we believe the result can be improved. 

\begin{Problem}
Which eigenvalue is much better than the minimum one?
\end{Problem}

\bibliographystyle{plain}
\bibliography{ref}

\end{document}